\documentclass[a4paper]{article}
\usepackage[margin=0.95in]{geometry}
\usepackage[dvipdfmx]{graphicx}
\usepackage[dvipdfmx]{color}

\usepackage{enumerate}
\usepackage{amsthm}
\usepackage{graphicx}
\usepackage{stmaryrd}
\usepackage{latexsym}
\usepackage{algorithm}
\usepackage{algorithmic}
\usepackage{amsmath}
\usepackage{amssymb}
\usepackage{url}

\def\unnumbered#1{%
\begin{trivlist}
\item[] {{#1}\kern0.5em}}
\def\endunnumbered{\end{trivlist}}

\def\myappenv#1#2{\begin{unnumbered}{{\bf #1~\ref{#2}.}}}
\def\endmyappenv{\end{unnumbered}}

\newtheorem{theorem}{Theorem}
\newtheorem{lemma}{Lemma}

\newtheorem{corollary}[theorem]{Corollary}

\newtheorem{remark}{Remark}

\usepackage[draft,mode=multiuser]{fixme}
\fxuseenvlayout{colorsig}
\fxusetargetlayout{color}
\FXRegisterAuthor{tt}{att}{TT}
\FXRegisterAuthor{si}{asi}{SI}
\FXRegisterAuthor{ha}{aha}{HA}
\fxsetface{env}{}
\fxsetface{inline}{\bfseries}

\renewcommand{\subsubsection}[1]{\noindent \textbf{#1}}

\newcommand{\set}[1]{\{ #1 \}}
\newcommand{\eps}{\varepsilon}
\newcommand{\sig}[1]{{\cal #1}}    
\newcommand{\pair}[1]{\langle #1\rangle}
\newcommand{\ceil}[1]{\lceil #1\rceil}
\newcommand{\name}{\emph}

\newcommand{\floor}[1]{\lfloor #1\rfloor}
\newcommand{\alg}[1]{{\sf #1}}

\newcommand{\polylog}{\mathrm{polylog}}

\newcommand{\etal}{\textit{et al.}}

\newcommand{\mylog}{\log}

\newcommand{\MTree}[1]{{\mathcal{MT}_{#1}}}
\newcommand{\PTree}[1]{{\mathcal{PT}_{#1}}}

\newcommand{\parent}{\mathit{parent}}
\newcommand{\rsl}{\mathit{sl_r}}
\newcommand{\level}{\mathit{level}}
\newcommand{\LCA}{\mathit{LCA}}
\newcommand{\Str}{\mathit{str}}
\newcommand{\LCP}{\mathit{LCP}}
\newcommand{\ld}{\mathit{ld}}
\newcommand{\Suf}{\mathit{Suf}}

\newcommand{\Pred}{\mathsf{Pred}}
\newcommand{\Succ}{\mathsf{Succ}}
\newcommand{\LPath}{\mathsf{LPS}}
\newcommand{\Root}{\mathit{root}}
\newcommand{\Insert}{\mathsf{Insert}}
\newcommand{\Delete}{\mathsf{Delete}}

\def\newblock{\hskip .11em plus .33em minus .07em} 

\newcommand{\fnbf}[1]{{\frac{(\mylog\mylog {#1})^2}{\mylog\mylog\mylog {#1}}}}			

\def\fn(#1){\fnbf{#1}}			

\begin{document}

\title{%
 Packed Compact Tries: A Fast and Efficient Data Structure for Online String Processing
} %
\author{
  Takuya~Takagi$^1$\quad
  Shunsuke~Inenaga$^2$\quad
  Kunihiko~Sadakane$^3$\quad\\
  Hiroki~Arimura$^1$\\
  {$^1$ Graduate School of IST, Hokkaido University, Japan}\\
  {\texttt{\{tkg,arim\}@ist.hokudai.ac.jp}}\\
  {$^2$ Department of Informatics, Kyushu University, Japan}\\
  {\texttt{inenaga@inf.kyushu-u.ac.jp}}\\
  {$^3$ Graduate School of Information Sci. and Tech., University of Tokyo, Japan}\\
  {\texttt{sada@mist.i.u-tokyo.ac.jp}}
}

\date{}


\maketitle

\begin{abstract}
In this paper, we present a new data structure
called the \emph{packed compact trie} (\emph{packed c-trie})
which stores a set $S$ of $k$ strings of total length $n$ 
in $n \log\sigma + O(k \log n)$ bits of space 
and supports fast pattern matching queries and updates,
where $\sigma$ is the size of an alphabet. 
Assume that $\alpha = \log_\sigma n$ letters are packed in a single machine word
on the standard word RAM model,
and let $f(k,n)$ denote the query and update times of
the dynamic predecessor/successor data structure of our choice
which stores $k$ integers from universe $[1,n]$ in $O(k \log n)$ bits of space.
Then, given a string of length $m$,
our packed c-tries support
pattern matching queries and insert/delete operations
in $O(\frac{m}{\alpha} f(k,n))$ worst-case time and 
in $O(\frac{m}{\alpha} + f(k,n))$ expected time.
Our experiments show that our packed c-tries are 
faster than the standard compact tries (a.k.a. Patricia trees)
on real data sets.
As an application of our packed c-trie, 
we show that the sparse suffix tree for a string of length $n$
over prefix codes with $k$ sampled positions,
such as evenly-spaced and word delimited sparse suffix trees, 
can be constructed online in
$O((\frac{n}{\alpha} + k) f(k,n))$ worst-case time and
$O(\frac{n}{\alpha} + kf(k,n))$ expected time
with $n \log \sigma + O(k \log n)$ bits of space.
When $k = O(\frac{n}{\alpha})$, 
by using the state-of-the-art dynamic predecessor/successor data structures,
we obtain sub-linear time construction algorithms 
using only $O(\frac{n}{\alpha})$ bits of space in both cases. 
We also discuss an application of our packed c-tries to online 
LZD factorization. 

\end{abstract}


\section{Introduction}





The trie for a set $S$ of strings of total length $n$
is a classical data structure 
which occupies $O(n\log n + n\log \sigma)$ bits of space
and allows for prefix search and insertion/deletion for a given string
of length $m$ in $O(m \log \sigma)$ time,
where $\sigma$ is the alphabet size.
The \emph{compact trie} for $S$, a.k.a. 
\emph{Patricia tree}~\cite{Morrison:1968}, is 
a path-compressed trie where the edges in every non-branching path
are merged into a single edge.
By representing each edge label by a pair of positions
in a string in $S$,
the compact trie can be stored in $n \log \sigma + O(k \log n)$ bits of space,
where $k$ is the number of strings in $S$,
retaining the same time efficiency for prefix search and insertion/deletion
for a given string.
Thus, compact tries have widely been used in numerous applications 
such as dynamic dictionary matching~\cite{Hon:Lam:Shah:Vitter:ISAAC:2009}, 
suffix trees~\cite{Weiner}, 
sparse suffix trees~\cite{Karkkainen:Ukkonen:1996}, 
external string indexes~\cite{Ferragina:Grossi:JACM:1999},
and grammar-based text compression~\cite{GotoBIT15}.


In this paper, we show how to accelerate 
prefix search queries and update operations of compact tries 
on the standard word RAM model with machine word size $w = \log n$,
still keeping $n \log \sigma + O(k \log n)$-bit space usage.
%
%
%
A basic idea is to use the \emph{packed string matching} 
approach~\cite{Benkiki:Bille:etal:FSTTCS:2011},
where $\alpha = \log_\sigma n$ consecutive letters are packed 
in a single word and can be manipulated in $O(1)$ time.
In this setting, we can read a given pattern $P$ of length $m$ in $O(\frac{m}{\alpha})$ time,
but, during the traversal of $P$ over a compact trie,
there can be at most $m$ branching nodes.
Thus, a na\"ive implementation of a compact trie
takes $O(\frac{m}{\log_\sigma n} + m \log \sigma) = O(m \log \sigma)$ time
even in the packed matching setting.
To overcome the above difficulty,
we propose how to quickly process long non-branching paths
using bit manipulations, 
and how to quickly process dense branching subtrees
using fast predecessor/successor queries and dictionary look-ups.
As a result, 
we obtain a new fast compact trie called the \emph{packed compact trie}
(\emph{packed c-trie}) for a dynamic set $S$ of strings,
which achieves the following efficiency:


\begin{theorem}[main result] \label{theo:main_result}
Let $f(k, n)$ be the query and update time complexity 
of an arbitrary dynamic predecessor/successor data structure
which occupies $O(k \log n)$ bits of space
for a dynamic set of $k$ integers from the universe $[1,n]$.
Then, our packed c-trie stores 
a set $S$ of $k$ strings of total length $n$
in $n \log \sigma + O(k \log n)$ bits of space
and supports prefix search and
insertion/deletion for a given string of length $m$ in 
$O(\frac{m}{\alpha}f(k,n))$ worst-case time
or in $O(\frac{m}{\alpha} + f(k,n))$ expected time.
\end{theorem}

If we employ Beame and Fich's data structure~\cite{Beame:Fich:JCSS:2002}
or Willard's y-fast trie~\cite{Willard:83} as the dynamic predecessor/successor
data structure, we obtain the following corollary:


\begin{corollary} \label{coro:beame_fich_willard}
There exists a packed c-trie for a dynamic set $S$ of strings 
which uses $n \log \sigma + O(k \log n)$ bits of space,
and supports prefix search and insert/delete operations for 
a given string of length $m$ in $O(\frac{m}{\alpha}\cdot\frac{\log \log k \log \log n }{\log\log\log n})$
worst-case time or in $O(\frac{m}{\alpha} + \log \log n)$ expected time.
\end{corollary}

An interesting feature of our packed c-trie
is that unlike most other (compact) tries,
our packed c-trie does \emph{not} maintain a dictionary or a search structure
for the children of each node.
Instead, we partition our c-trie into $\ceil{h / \alpha}$ levels,
where $h$ is the length of the longest string in $S$.
Then each subtree of height $\alpha$, called a \emph{micro} c-trie,
maintains a predecessor/successor dictionary that processes
prefix search inside the micro c-trie.
A similar technique is used
in the \emph{linked dynamic trie}~\cite{JanssonSS15}, which is
an \emph{uncompact} trie for a dynamic set of strings.

Our experiments show that our packed c-tries are faster 
than Patricia trees for both construction and prefix search
in almost all data sets we tested.


			
We also show two applications to our packed c-tries.
The first application is online construction of 
\emph{evenly sparse suffix trees}~\cite{Karkkainen:Ukkonen:1996},
\emph{word suffix trees}~\cite{Inenaga:Takeda:CPM:2006} 
and its extension~\cite{Uemura:Arimura:CPM:2011}.
The existing algorithms for these sparse suffix trees take 
$O(n \log \sigma)$ worst-case time using $n \log \sigma + O(k \log n)$ bits of space,
where $k$ is the number of suffixes stored in the output sparse suffix tree.
Using our packed c-tries,
we achieve $O((\frac{n}{\alpha}+k) \frac{\log \log k \log \log n}{\log \log \log n})$ 
worst-case construction time
and $O(\frac{n}{\alpha} + k \log \log n)$ expected construction time.
The former is sublinear in $n$ when $k = O(\frac{n}{\alpha})$ and $\sigma = \polylog(n)$,
the latter is sublinear in $n$ when $k = o(\frac{n}{\log \log n})$ and $\sigma = \polylog(n)$.
To achieve these results,
we show that in our packed c-trie,
prefix searches and insertion operations
can be started not only from the root but from \emph{any} node.
This capability is necessary for online sparse suffix tree construction,
since during the suffix link traversal we have to insert new leaves
from non-root internal nodes.

The second application is 
online computation of the \emph{LZ-Double factorization}~\cite{GotoBIT15} 
(\emph{LZDF}), a state-of-the-art online grammar-based text compressor.
Goto et al.~\cite{GotoBIT15} presented a Patricia-tree based algorithm
which computes the LZDF of a given string $T$ of length $n$ in 
$O(k(M + \min\{k, M\} \log \sigma))$ worst-case time using $O(n \log \sigma)$
bits of space, where $k \leq n$ is the number of factors 
and $M \leq n$ is the length of the longest factor.
Using our packed c-tries,
we achieve a good expected performance with 
$O(k(\frac{M}{\alpha} + f(k, n)))$ time for LZDF.

All the proofs omitted due to lack of space can be found in Appendix.

\subsubsection{Related work.}
Belazzougui et al.~\cite{Belazzougui:SPIRE:2010} proposed 
a \emph{randomized} compact trie called 
the \emph{signed dynamic z-fast trie}, which stores
a dynamic set $S$ of $k$ strings in $n \log\sigma + O(k \log n)$ 
bits of space.
Given a string of length $m$,
the signed dynamic z-fast trie supports prefix search 
in $O(\frac{m}{\alpha} + \log m)$ worst-case time 
\emph{only with high probability},
and supports insert/delete operations 
in $O(\frac{m}{\alpha} + \log m)$ expected time
\emph{only with high probability}.\footnote{The $O(\log m)$ expected bound 
for insertion/deletion stated in~\cite{Belazzougui:SPIRE:2010} assumes that the prefix search for the string has already been performed.}
On the other hand, our packed c-trie always return the correct answer
for prefix search, and always insert/delete a given string correctly,
in the bounds stated in Theorem~\ref{theo:main_result} 
and Corollary~\ref{coro:beame_fich_willard}.

Andersson and Thorup~\cite{AnderssonT07}
proposed the \emph{exponential search tree}
which uses $n \log \sigma + O(k \log n)$ bits of space,
and supports prefix search and insert/delete operations 
in $O(m + \sqrt{\frac{\log k}{\log \log k}})$ worst-case time.
Each node $v$ of the exponential search tree stores 
a constant-time look-up dictionary for some children of $v$
and a dynamic predecessor/successor for the other children of $v$.
This implies that given a string of length $m$,
at most $m$ nodes in the search path for the string
must be processed one by one,
and hence packing $\alpha = \log_\sigma n$ letters in a single word
does not speed-up prefix searches or updates 
on the exponential search tree.

Fischer and Gawrychowski~\cite{FischerG2015} proposed 
the \emph{wexponential search tree},
which uses $n \log \sigma + O(k \log n)$ bits
of space, and supports prefix search and insert/delete operations
in $O(m + \frac{(\log \log \sigma)^2}{\log \log \log \sigma})$ 
worst-case time.
When $\sigma = \polylog(n)$,
our packed c-trie achieves the worst-case bound
$O(m\frac{\log \sigma \log \log k \log \log n}{\log n\log \log \log n})
= O(m\frac{(\log \log n)^2}{\log n\log \log \log n})
= O(o(1)m)$, 
whereas the wexponential search tree requires $O(m + \frac{(\log \log \log n)^2}{\log \log \log \log n})$ time\footnote{For sufficiently long patterns of length $m = \Theta(n)$, our packed c-trie achieves worst-case \emph{sublinear} $o(n)$ time while the wexponential search tree requires $O(n)$ time.}.

\section{Preliminaries}

Let $\Sigma$ be the alphabet of size $\sigma$.
An element of $\Sigma^*$ is called a string.
For any string $X$ of length $n$, $|X|$ denotes its length, namely $|X| = n$.
We denote the empty string by $\eps$.
For any $1 \leq i \leq n$, $X[i]$ denotes the $i$th character of $X$. 
For any $1 \le i \le j \le |X|$, $X[i,j]$ denotes the substring $X[i] \cdots X[j]$.
For convenience, $X[i,j] = \eps$ for $i > j$.
For any strings $X, Y$, $\LCP(X, Y)$ denotes 
the longest common prefix of $X$ and $Y$.


Throughout this paper, the base of the logarithms will be 2, unless otherwise stated. 
For any integers $i\le j$, $[i,j]$ denotes the interval $\set{i, i+1, \ldots, j}$. 
Our model of computation is the standard 
word RAM of word size $w = \mylog n$ bits.
For simplicity, we assume that $w$ is a multiple of $\log\sigma$,
so $\alpha = \log_\sigma n$ letters are packed in a single word.
Since we can read $w$ bits in constant time, 
we can read and process $\alpha$ consecutive letters in constant time. 

Let $S = \{X_1, \ldots, X_k\}$ be a set of $k$ non-empty strings
of total length $n$.
In this paper, we consider dynamic data structures for $S$
which allows us fast prefix searches of given patterns over strings in $S$,
and fast insertion/deletion of strings to/from $S$.


Suppose $S$ is prefix-free.
The \emph{trie} of $S$ is a tree such that each edge is labeled by a single letter,
the labels of the out-going edges of each node are distinct,
and there is a one-to-one correspondence between the strings in $S$
and the leaves, namely, 
for each $X_i \in S$ there exists a unique path from the root to a leaf
that spells out $X_i$.

The \emph{compact trie} $\sig T_S$ of $S$
is a path-compressed trie obtained by contracting
a non-branching path into a single edge.
Namely, in $\sig T_S$,
each edge is labeled by a non-empty substring of $T$,
each internal node has at least two children,
the out-going edges from each node begin with distinct letters,
and each edge label $x$ is encoded by a triple $\pair{i,a,b}$
such that $x = X_i[a,b]$ for some $1 \leq i \leq k$ and 
$1 \leq a \leq b \leq |X_i|$.
The \emph{length} of an edge $e$, denoted $|e|$, 
is the length of its label string.
Let $\Root(\sig T_S)$ denote the root of the compact trie $\sig T_S$.
For any node $v$, let $\parent(v)$ denotes its parent.
For convenience, let $\bot$ be an auxiliary node
such that $\parent(\Root(\sig T_S)) = \bot$.
We also assume that the edge from $\bot$ to $\Root(\sig T_S)$ is 
labeled by an arbitrary letter.
For any node $v$, let $\Str(v)$ denotes the string obtained by
concatenating the edge labels from the root to $v$.
We assume that each node $v$ stores $|\Str(v)|$.

Let $s$ be a prefix of any string in $S$.
Let $v$ be the shallowest node of $\sig T_S$ such that $s$ is a suffix of $\Str(v)$ (notice $s$ can be equal to $\Str(v)$), and let $u = \parent(v)$.
The \emph{locus} of string $s$ in the compact trie $\sig T_S$
is a pair $\phi = (e, h)$,
where $e$ is the edge from $u$ to $v$
and $h$~($1 \leq h \leq |e|$) is the offset from $u$, 
namely, $h = |s| - |\Str(u)|$.\footnote{In the literature the locus is represented by $(u, c, h)$ where $c$ is the first letter of the label of $e$. Since our packed c-trie does not maintain a search structure for branches, we represent the locus directly on $e$.}
We extend the $\Str$ function to locus $\phi$, 
so that $\Str(\phi) = s$.
The \emph{string depth} of locus $\phi$ is $d(\phi) = |\Str(\phi)|$.
We say that a string $P$ is \emph{recognized} by $\sig T_S$
iff there is a locus $\phi$ with $\Str(\phi) = P$.

Our input is a dynamic set of strings 
which allows for insertion and deletion of strings.
We thus consider the following query and operations on dynamic compact tries.

\begin{itemize}

\item $\LPath(\phi, P)$: 
Given a locus in $\sig T_S$ and a pattern string $P$,
it returns the locus $\hat{\phi}$ of string $\Str(\phi)Q$ in $\sig T_S$,
where $Q$ is the longest prefix of $P$ for which $\Str(\phi)Q$
is recognized by $\sig T_S$.
When $\phi = ((\bot, \Root(\sig T_S)), 1)$, 
then the query is known as the \emph{longest prefix search} for 
the pattern $P$ in the compact trie.

\item $\Insert(\phi, X)$: Given a locus $\phi$ in $\sig T_S$
and a string $X$,
it inserts a new leaf which corresponds to
a new string $\Str(\phi) X \in S$ into the compact trie,
from the given locus $\phi$.
When there is no node at the locus $\hat{\phi} = \LPath(\phi,X)$,
then a new node is created at $\hat{\phi}$ as the parent of the leaf.
When $\phi = ((\bot, \Root(\sig T_S)), 1)$, then this is standard
insertion of string $X$ to $\sig T_S$.

\item $\Delete(X_i)$: Given a string $X_i \in S$,
it deletes the leaf node $\ell_i$.
If the out-degree of the parent $v$ of $\ell_i$ becomes 1
after the deletion of $\ell_i$, then the in-coming and out-going 
edges of $v$ are merged into a single edge,
and $v$ is also deleted.






\end{itemize}



For a dynamic set $I \subseteq [1,n]$ of $k$ integers of $w = \mylog n$ bits each,
\name{dynamic predecessor data structures} 
(e.g., \cite{Beame:Fich:JCSS:2002,Belazzougui:SPIRE:2010,Willard:84}) 
efficiently support predecessor query
$\Pred(X) = \max(\{Y \in I \mid Y \leq X\} \cup \{0\})$,
successor query $\Succ(X) = \min(\{Y \in I \mid Y \leq X\} \cup \{n+1\})$,
and insert/delete operations for $I$.
Let $f(k, n)$ be the time complexity of 
for predecessor/successor queries and insert/delete operations
of an arbitrary dynamic predecessor/successor data structure
which occupies $O(k \log n)$ bits of space.
Beame and Fich's data structure~\cite{Beame:Fich:JCSS:2002} 
achieves $f(k, n) = O(\frac{(\log\log k)(\log \log n)}{\log \log \log n})$ 
worst-case time,
while Willard's Y-fast trie~\cite{Willard:83} achieves 
$f(k, n) = O(\log \log n)$ expected time.

\section{Packed dynamic compact tries}
\label{sec:algo}

In this section, we present our new dynamic compact tries
called the \emph{packed dynamic compact tries} 
(\emph{packed c-tries}) for a dynamic set $S = \{X_1, \ldots, X_k \}$ of 
$k$ strings of total length $n$,
which achieves the main result in Theorem~\ref{theo:main_result}.
In the sequel, a string $X \in \Sigma^*$ is called \emph{short} 
if $|X| \leq \alpha = \log_\sigma n$,
and is called \emph{long} if $|X| > \alpha$.


\subsection{Micro dynamic compact tries for short strings}
\label{subsec:small}

In this subsection, we present our data structure storing short strings.
Our input is a dynamic set $S = \{X_1, \ldots, X_k\}$ of 
$k$ strings of total length $n$,
such that $|X_i| \leq \alpha = \mylog_\sigma n$ for every $1 \leq i \leq k$.
Hence it holds that $k \leq \sigma^\alpha = n$.
For simplicity, we assume for now that $|X_i| = \alpha$
for every $1 \leq i \leq k$.
The general case where $S$ contains strings shorter than $\alpha$ will be 
explained later in Remark~\ref{rem:shorter_then_alpha}.

The dynamic data structure for short strings, 
called a \emph{micro c-trie} and denoted by $\MTree{S}$,  
consists of the following components:
\begin{itemize}

\item A dynamic compact trie of height exactly $\alpha$ storing the set $S$.
Let $\mathcal{N}$ be the set of internal nodes, 
and let $\mathcal{L} = \{\ell_1, \ldots, \ell_k\}$ be 
the set of $k$ leaves 
such that $\ell_i$ corresponds to $X_i$ for $1 \leq i \leq k$.
Since every internal node is branching, $|\mathcal{N}| \leq k-1$.
Every node $v$ of $\MTree{S}$ explicitly stores the string $\Str(v)$
using $\log n$ bits.
This implies that we can identify $v$ with $\Str(v)$.
Overall, this compact trie requires 
$n \log \sigma + O(k \log n)$ bits of space (including $S$).

\item A dynamic predecessor/successor data structure $\mathcal{D}$
which stores the set $S = \{X_1, \ldots, X_k\}$ of strings
in $O(k \log n)$ bits of space,
where each $X_i$ is regarded as a $\log n$-bit integer.
$\mathcal{D}$ supports predecessor/successor queries and
insert/delete operations in $f(k, n)$ time each.


\end{itemize}
It is evident that the micro c-trie requires 
$n \log \sigma + O(k \log n)$ bits of total space.



\begin{lemma}
\label{lem:LCA_micro_tree}
For any nodes $u$ and $v$ of the micro c-trie $\MTree{S}$,
we can compute the lowest common ancestor $\LCA(u, v)$ of $u$ and $v$
in $O(1)$ time.
\end{lemma}

\begin{proof}
We pad $\Str(u)$ and/or $\Str(v)$
with an arbitrary letter $c$ if necessary.
Namely, if $|\Str(u)| = \alpha$ then let $P = \Str(u)$, and
if $|\Str(u)| < \alpha$ then let $P = \Str(u) c \cdots c \in \Sigma^\alpha$. 
Similarly, 
if $|\Str(v)| = \alpha$ then let $Q = \Str(v)$, and
if $|\Str(v)| < \alpha$ then let $Q = \Str(v) c \cdots c \in \Sigma^\alpha$. 
We compute the most significant bit (msb) of the XOR of 
the bit representations of $P$ and $Q$.
Let $b$ the bit position of the msb,
and let $z = (b-1) / \log \sigma$.
W.l.o.g. assume $|\Str(u)| \leq |\Str(v)|$.
\begin{enumerate}
\item If $z < \Str(u)$, then 
$\Str(u)[1,z] = \LCP(\Str(u), \Str(v))$.
In this case, there exists a branching node $y$ such that $\Str(y) = \Str(u)[1,z]$,
and hence $\LCA(u, v) = y$.

\item If $z \geq \Str(u)$,
then $\Str(u) = \LCP(\Str(u), \Str(v))$,
and hence $u = \LCA(u, v)$.
\end{enumerate}

Since each of $P$ and $Q$ is stored in a single machine word,
we can compute the XOR of $P$ and $Q$ in $O(1)$ time.
The msb can be computed in $O(1)$ time using 
the technique of Fredman and Willard~\cite{FredmanW93}.
This completes the proof.
\end{proof}

\begin{theorem}
\label{thm:small:opr:LCP}
The micro c-trie $\MTree{S}$
supports $\LPath(\phi, X)$ queries in $O(f(k, n))$ time.
\end{theorem}

\begin{proof}
Our algorithm for computing $\hat{\phi} = \LPath(\phi, X)$ consists of 
the two following steps:

First, we compute the string depth $d = d(\phi) \in [0,\alpha]$.
Let $P = \Str(\phi)X[1..\alpha-d(v)]$ be the prefix of $\Str(v)X$ of length $\alpha$.
Observe $d = \max\{|\LCP(P, \Pred(P))|, |\LCP(P, \Succ(P))|\}$.
Given $P$, we compute $\Pred(P)$ and $\Succ(P)$ in $O(f(k, n))$ time.
Then, $|\LCP(P, \Pred(P))|$ can be computed in 
$O(1)$ time by computing the msb of the XOR of 
the bit representations of $P$ and $\Pred(P)$, as in Lemma~\ref{lem:LCA_micro_tree}.
$|\LCP(P, \Succ(P))|$ can be computed analogously,
and thus, $d = d(\phi)$ can be computed in $O(f(k, n))$ time.

Second, we locate $e = (u, v)$.
See also Fig.~\ref{fig:small:case}. 
Let $Z = P[1, d]$.
Let $\mathit{LB} = Z c_1 \cdots c_1 \in \Sigma^\alpha$ 
and $\mathit{UB} = Z c_\sigma \cdots c_\sigma \in \Sigma^\alpha$ 
be the lexicographically least and greatest strings of length $\alpha$ with prefix $Z$,
respectively. 
To locate $u$ in $\MTree{S}$, we find 
the leftmost and rightmost leaves $X_L$ and $X_R$ below $\phi$ 
by $X_L = \Succ(\mathit{LB})$ and $X_R = \Pred(\mathit{UB})$. 
Then, the lower one of 
$\LCA(X_{L-1}, X_L)$ and 
$\LCA(X_R, X_{R+1})$ is the origin node $u$ of $e$.
The destination node $v$ is $\LCA(X_L, X_R)$.
These LCAs can be computed in $O(1)$ time by Lemma~\ref{lem:LCA_micro_tree}.
Finally we obtain 
$\phi = ((u, v), d - |\Str(u)|)$.
Overall, this step takes $O(f(k, n))$ time.
\end{proof}

\begin{figure}[t]
\begin{minipage}[t]{0.45\textwidth}
\begin{center}
\includegraphics[scale=0.45]{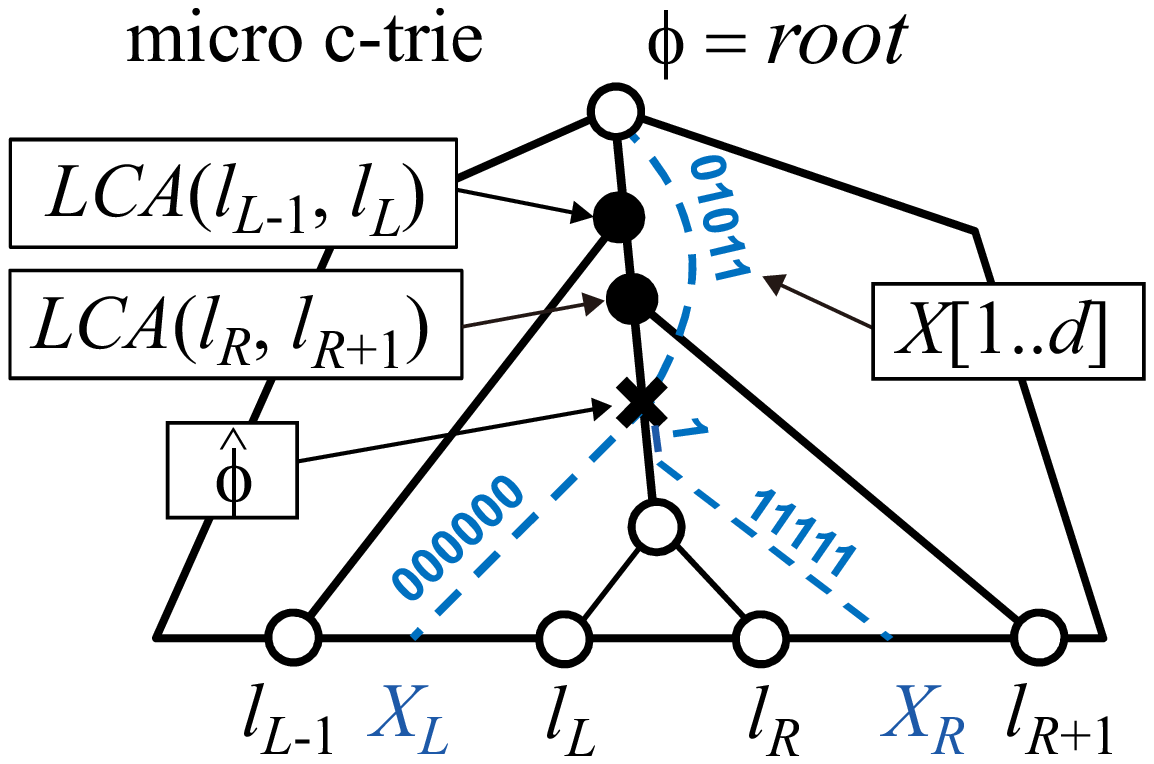}
\caption{
Given the initial locus $\phi$ (which is on the root in this figure)
and query pattern $P = 01011010110$,
the algorithm of Theorem~\ref{thm:small:opr:LCP} answers 
the $\LPath(\phi, P)$ query on the micro c-trie
as in this figure.
The answer to the query is the locus $\hat{\phi}$
for $P[1..5] = 01011$.
}
\label{fig:small:case}
\end{center}
\end{minipage}
\hfill
\begin{minipage}[t]{0.5\textwidth}
\begin{center}
\includegraphics[scale=0.4]{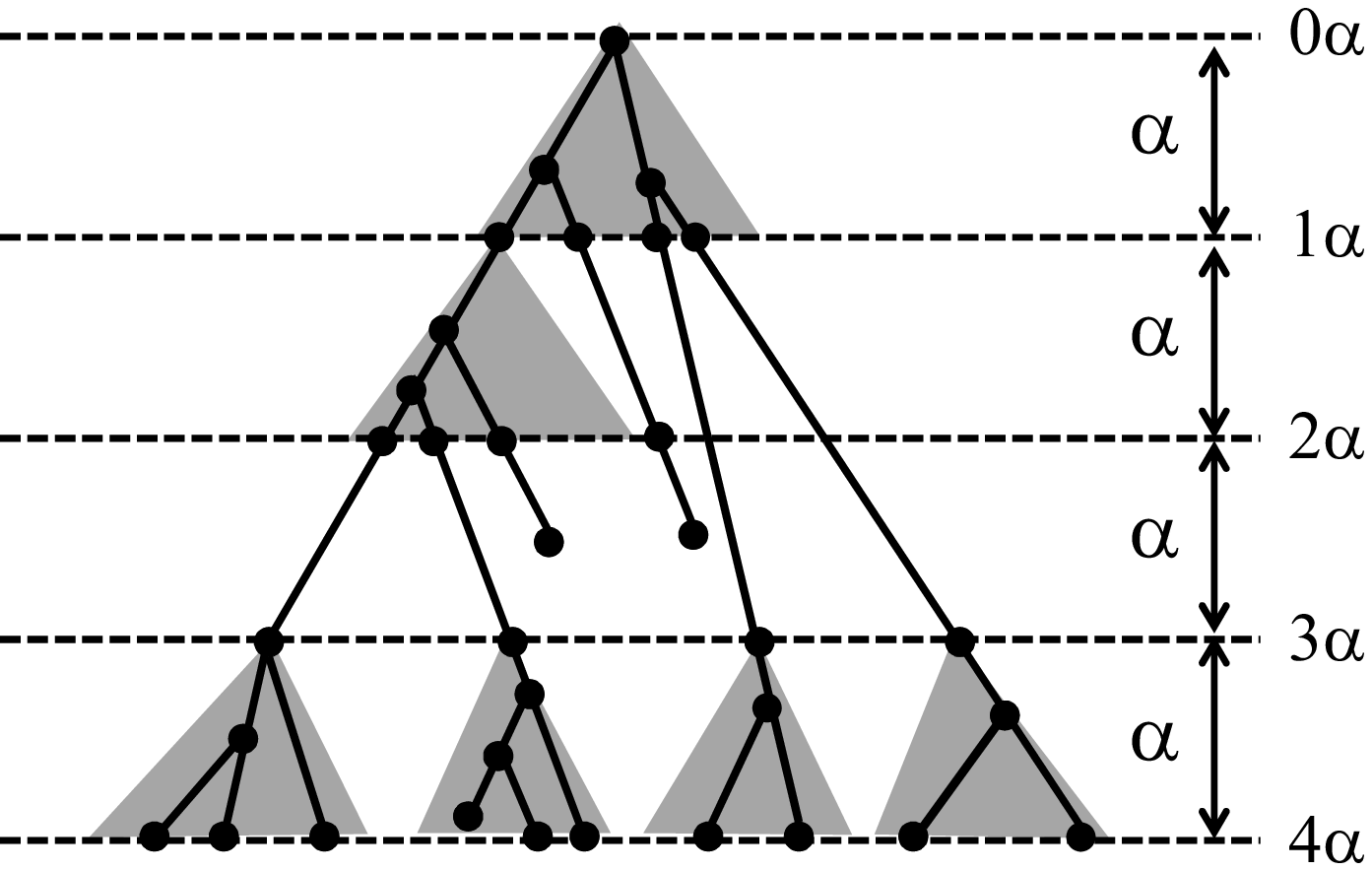}
\caption{Micro-trie decomposition:
The packed c-trie is decomposed into a number of 
micro c-tries (gray rectangles) each of which is 
of height $\alpha = \log_\sigma n$. Each micro-trie is equipped with a dynamic predecessor/successor data structure.}
\label{fig:large:case}
\end{center}
\end{minipage}
\end{figure}

Next, we explain how to support $\Insert(\phi, X)$ and $\Delete(X)$ operations.
\begin{lemma}
\label{lem:small:semidyn}
The micro c-trie $\MTree{S}$
supports $\Insert(\phi, X)$ and $\Delete(X)$ operations 
in $O(f(k, n))$ time.
We assume that $d(\phi) + |X| \le \alpha$ so that the height of 
the micro compact trie will always be kept within $\alpha$. 
\end{lemma}

\begin{proof}
We show how to support $\Insert(\phi, X)$ in $O(f(k, n))$ time.
Initially $S = \emptyset$, 
the micro compact trie $\MTree{S}$ consists only of $root(\MTree{S})$, 
and predecessor/successor dictionary $\sig D$ contains no elements. 
When the first string $X$ is inserted to $S$, 
then we create a leaf below the root and insert $X$
to $\sig D$.
Suppose that the data structure maintains a string set $S$ with $|S|\ge 1$. 
To insert a string $X$ from the given locus $\phi$,
we first conduct the $\LPath(\phi, X)$ query of Theorem~\ref{thm:small:opr:LCP},
and let $\hat{\phi} = (e, h)$ be the answer to the query.
If $h = |e|$, then we simply insert a new leaf $\ell$ from
the destination node of $e$.
Otherwise, we split $e$ at $\hat{\phi}$ and
create a new node $v$ there as the parent of the new leaf, 
such that $\Str(v) = \Str(\hat{\phi})$.
The rest is the same as in the former case.
After the new leaf is inserted,
we insert $\Str(\phi)X$ to $\sig D$ in $O(f(k, n))$ time. 

We can support $\Delete(X)$ as follows.
Let $\ell_i$ be the leaf representing $X_i = X \in S$.
If $\ell_i$ is a child of the root,
then we simply delete $\ell_i$.
Otherwise, we employ the following trick:
For each leaf $\ell$ in the micro c-trie,
we maintain the rank $r(\ell)$ such that $r(\ell) = t$ iff 
$\ell$ is the $t$-th inserted leaf to the micro c-trie.
Let $\ell_j$ be any sibling of $\ell_i$ with $j \neq i$.
If $r(\ell_i) > r(\ell_j)$, then no edge labels 
in the path $\mathcal{P}$ from the root to $\parent(\ell_i)$
refer to positions in $X_i$,
and hence we simply delete $\ell_i$ from the tree and $X_i$ from $S$.
If $r(\ell_i) < r(\ell_j)$, then some edge labels
in path $\mathcal{P}$ refer to positions in $X_i$.
The important observation is that, by the way we insert strings
using $\Insert$ queries above,
no edge labels in $\mathcal{P}$ refer to string $X_j$.
Now, we \emph{swap} the strings $X_i$ and $X_j$,
and delete $\ell_j$ from the trie and $X_j$ from $S$
(e.g, if $X_i = \mathtt{aabb}$ and $X_j = \mathtt{aaab}$,
then we swap them as $X_i = \mathtt{aaab}$ and $X_j = \mathtt{aabb}$,
and delete $X_j = \mathtt{aabb}$).
We can swap these strings in $O(1)$ time since they are of length $\alpha$.
When the rank value reaches $2n$ after the $2n$-th insertion,
then we re-label the ranks of all $k$ existing leaves from $1$ to $k$
in $O(n)$ time using a bucket sort.
Since $k \leq n$, the amortized cost for the re-labeling is constant.
Thus, the total time cost for $\Delete(X)$ is $O(f(k, n))$.
\end{proof}

\begin{remark}
\label{rem:shorter_then_alpha}
When $d(\phi) + |X| < \alpha$, then we can support 
$\Insert(\phi, X)$ and $\LPath(\phi, X)$ as follows.
When inserting $X$, we pad $X$ with a special letter $\$$
which does not appear in $S$.
Namely, we perform $\Insert(\phi, X)$ operation with
$X' = X \$^{\alpha-d(\phi)-|X|}$.
When computing $\LPath(\phi, X)$, 
we pad $X$ with another special letter $\# \neq \$$
which does not appear in $S$.
Namely, we perform $\LPath(\phi, X'')$ query
with $X' = X \#^{\alpha-d(\phi)-|X|}$.
This gives us the correct locus for $\LPath(\phi, X)$.
\end{remark}

\subsection{Packed dynamic compact tries for long strings}
\label{subsec:large}

In this subsection, 
we present the \emph{packed dynamic compact trie} (\emph{packed c-trie}) 
$\PTree{S}$ for a set $S$ of variable-length strings of length at most $O(2^w) = O(n)$.


\subsubsection{Micro trie decomposition.}
We decompose $\PTree{S}$ into a number of micro c-tries:
Let $h > \alpha$ be the length of the longest string in $S$.
We categorize the nodes of $\PTree{S}$ into $\ceil{h / \alpha} + 1$ 
levels:
We say that a node of $\PTree{S}$ is at level 
$i$~($0 \leq i \leq \ceil{h / \alpha}$)
iff $|\Str(v)| \in [i\alpha, (i+1)\alpha-1]$.
The level of a node $v$ is denoted by $\level(v)$.
A locus $\phi$ of $\PTree{S}$ is called a \emph{boundary}
iff $d(\phi)$ is a multiple of $\alpha$. 
Consider any path from $\Root(\PTree{S})$ to a leaf,
and assume that there is no node at some boundary $k \alpha$ on this path.
We create an auxiliary node at that boundary on this path, 
iff there is at least one non-auxiliary (i.e., original)
node at level $i-1$ or $i+1$ on this path.
Let $\mathcal{BN}$ denote the set of nodes at the boundaries,
called the \emph{boundary nodes}.
For each boundary node $v \in \mathcal{BN}$,
we create a micro compact trie $\sig{MT}$ 
whose root $\Root(\sig{MT})$ is $v$,
internal nodes are all descendants $u$ of $v$ with $\level(u) = \level(v)$,
and leaves are all boundary descendants $\ell$ of $v$ 
with $\level(\ell) = \level(v) + 1$.
Notice that each boundary node 
is the root of a micro c-trie at its level
and is also a leaf of a micro c-trie at the previous level.

\begin{lemma}
\label{lem:space:large:case}
The packed c-trie $\PTree{S}$ for a prefix-free set $S$ of $k$ 
strings requires $n \log \sigma + O(k \log n)$ bits of space.
\end{lemma}

\begin{proof}
Firstly, we show the number of auxiliary boundary nodes
in $\PTree{S}$.
At most 2 auxiliary boundary nodes are created on 
each \emph{original} edge of $\PTree{S}$.
Since there are at most $2k-2$ original edges, 
the total number of auxiliary boundary nodes is at most $4k-4$.

Since there are at most $2k-1$ original nodes in $\PTree{S}$,
the total number of all nodes in $\PTree{S}$ is bounded by $6k-5$.
Clearly, the total number of \emph{short} strings of length at most $\alpha$ 
maintained by the micro c-tries is bounded by the number of all nodes
in $\PTree{S}$, which is $6k - 5$.
Hence, the total space of the packed c-trie $\PTree{S}$ is $n \log \sigma + O(k \log n)$ bits.
\end{proof}

For any locus $\phi$ on the packed c-trie $\PTree{S}$,
$\ld(\phi)$ denotes the local string depth of $\phi$
in the micro c-trie $\sig{MT}$ that contains $\phi$.
Namely, if $\Root(\sig{MT}) = v$, 
the parent of $u$ in $\PTree{S}$ is $u$,
and $e = (u, v)$, then $\ld(\phi) = d(\phi) - d((e, |e|))$.
Prefix search queries
and insert/delete operations can be efficiently 
supported by our packed c-trie, as follows.




\begin{lemma}
\label{lem:large:LCP}
The packed c-trie $\PTree{S}$
supports $\LPath(\phi, P)$ query in $O(\frac{m}{\alpha} f(k))$
worst-case time, where $m = |P| > \alpha$.
\end{lemma}

\begin{proof}
If $m + \ld(\phi) \leq \alpha$, then
the bound immediately follows from Theorem~\ref{thm:small:opr:LCP}.
Now assume $m + \ld(\phi) > \alpha$,
and let $q = \alpha - \ld(\phi)+1$.
We factorize $P$ into $h+1$ blocks as 
$p_0 = P[1,q-1]$, 
$p_1 = P[q,q+\alpha-1]$, \ldots, 
$p_{h-1} = P[q+(h-1)\alpha, q+h\alpha-1]$, and
$p_{h} = P[q+h\alpha, m]$,
where $1 \leq |p_0| \leq \alpha$,
$|p_i| = \alpha$ for $1 \leq i \leq h-1$,
and $1 \leq |p_{h}| \leq \alpha$.
Note that each block can be computed in $O(1)$ time by standard bit operations.
If there is a mismatch in $p_0$, we are done.
Otherwise, 
for each $i$ in increasing order from $1$ to $h$,
we conduct $\LPath(\gamma, p_i)$ query from 
the root $\gamma$ of the corresponding micro c-trie at each level
of the corresponding path starting from $\phi$.
This continues until either we find the first mismatch for some $i$,
or we find complete matches for all $i$'s.
Each $\LPath$ query with each micro c-trie
takes $O(f(k, n))$ time by Theorem~\ref{thm:small:opr:LCP}.
Since $h = O(\frac{m}{\alpha})$, it takes a total of $O(\frac{m}{\alpha}f(k, n))$ time.
\end{proof}

\begin{lemma}
The packed c-trie $\PTree{S}$
supports $\Insert(\phi, X)$ and $\Delete(X)$ operations
in $O(\frac{m}{\alpha} f(k, n))$ worst-case time,
where $m = |X| > \alpha$.
\end{lemma}

\begin{proof}
To conduct $\Insert(\phi, X)$ operation,
we first perform $\LPath(\phi, X)$ query in $O(\frac{m}{\alpha}f(k, n))$ time
using Lemma~\ref{lem:large:LCP}.
Let $x_0, \ldots, x_h$ be the factorization of $X$ w.r.t. $\phi$,
and let $x_j$ be the block of the factorization 
which contains the first mismatch.
Then, we conduct $\Insert(\gamma, x_j)$ operation
on the corresponding micro c-trie, where $\gamma$ is its root.
This takes $O(f(k, n))$ time by Lemma~\ref{lem:small:semidyn}.
If $j = h$ (i.e. $x_j$ is the last block in the factorization of $X$), 
then we are done.
Otherwise, we create a new edge whose label is
$x_{j}' x_{j+1} \cdots x_k$,
where $x_{j}'$ is the suffix of $X_j$ which begins with the 
mismatched position, leading to the new leaf $\ell$.
We create a new boundary node if necessary.
These operations take $O(1)$ time each.
Hence, $\Insert(\phi, X)$ is supported in $O(\frac{m}{\alpha}f(k, n))$ 
total time.

For $\Delete(X)$ operation,
we perform the operation of Lemma~\ref{lem:small:semidyn}
for each micro c-trie in the path from the root to the leaf representing $X$.
Since there are at most $\frac{m}{\alpha}$ such micro c-tries,
$\Delete(X)$ can be supported in $O(\frac{m}{\alpha}f(k, n))$ total time.
\end{proof}


\subsubsection{Speeding-up with hashing.}
By augmenting each micro c-trie with a hash table
storing the short strings in the trie,
we can achieve a good expected performance, as follows:
\begin{lemma} \label{lem:large:hash}
The packed c-trie $\PTree{S}$ with hashing
supports $\LPath(\phi, X)$ query, $\Insert(\phi, X)$ 
and $\Delete(X)$ operations
in $O(\frac{m}{\alpha} + f(k, n))$ expected time.
\end{lemma}

\begin{proof}
Let $\sig MT$ be any micro c-trie in the packed c-trie $\PTree{S}$,
and $M$ the set of strings maintained by $\sig MT$ 
each being of length at most $\alpha$.
We store all strings of $M$ in a hash table associated to $\sig MT$,
which supports look-ups, insertions and deletions in $O(1)$ expected time.

Let $x_0, \ldots, x_h$ be the factorization of $X$ w.r.t. $\phi$.
To perform $\LPath(\phi, X)$,
we ask if $\Str(\phi)x_0$ is 
in the hash table of the corresponding micro c-trie.
If the answer is no, the first mismatch occurs in $x_0$, 
and the rest is the same as 
in Lemma~\ref{lem:large:LCP}.
If the answer is yes, then for each $i$ from $1$ to $h$ in increasing order,
we ask if $x_i$ is in the hash table of the corresponding micro c-trie,
until we receive the first no with some $i$ or 
we receive yes for all $i$'s.
In the latter case, we are done.
In the former case, we perform $\LPath$ query with $x_i$ from the root of the 
corresponding micro c-trie.
Since we perform at most one $\LPath$ query 
and $O(\frac{m}{\alpha})$ look-ups for hash tables, 
it takes $O(\frac{m}{\alpha} + f(k, n))$ expected time.
$O(\frac{m}{\alpha} + f(k, n))$ expected time bounds for 
$\Insert(\phi, X)$ and $\Delete(X)$ immediately follow
from the above arguments.
\end{proof}

\section{Applications to online string processing}
\label{sec:appl}

In this section, we present two applications 
of our packed c-tries for online string processing.



\subsubsection{Online sparse suffix tree construction.}
The \emph{suffix tree}~\cite{Weiner} of a string $T$ of length $n$ is 
a compact trie which stores all $n$ suffixes 
$\Suf(T) = \{T[i..n] \mid 1 \leq i \leq n\}$ of $T$
in $n \log \sigma + O(n \log n)$ bits.
A \emph{sparse suffix tree} for a set $K \subseteq [1, n]$ 
of \emph{sampled positions} of $T$ is a compact trie
which stores only the subset $\{T[i..n] \mid i \in K\}$
of the suffixes of $T$ beginning at the sampled positions in $K$.
It is known that if the set $K$ of sampled positions
satisfy some properties (e.g., every $r$ positions for some fixed $r > 1$
or the positions immediately after the word delimiters),
the sparse suffix tree can be constructed 
in an online manner in $O(n \log \sigma)$ time
and $n \log \sigma + O(n \log n)$ bits of space~\cite{Karkkainen:Ukkonen:1996,Inenaga:Takeda:CPM:2006,Uemura:Arimura:CPM:2011}.

In this section, we show our packed c-tries can be used to 
speed up online construction and pattern matching for these sparse suffix trees.
We insert the suffixes in increasing order 
of their beginning positions (sampled positions) to the packed c-trie.
There, each input string $X$ to $\Insert(\phi, X)$ operation is 
given as a pair $(i, j)$ of positions in $T$ such that $X =T[i,j]$.
In this case, $\Insert$ operation can be processed more quickly
than in Lemma~\ref{lem:large:LCP}, as follows.

\begin{lemma} \label{lem:faster_insert}
Given a pair $(i, j)$ of positions in $T$ such that $X = T[i,j]$,
we can support $\Insert(\phi, X)$ operation in $O(\frac{q}{\alpha} f(k,n))$ 
worst-case time or $O(\frac{q}{\alpha} + f(k,n))$ expected time,
where $q$ is the length of the longest prefix of $X$
that can be spelled out from the locus $\phi$.
\end{lemma}



\newcommand{\myceil}[1]{{(#1)}}

\begin{theorem} \label{thm:sparse_suffix_tree}
Using our packed c-tries,
we can construct in an online manner
the sparse suffix trees of 
\cite{Inenaga:Takeda:CPM:2006,Karkkainen:Ukkonen:1996,Uemura:Arimura:CPM:2011} for a given text $T$ of length $n$
in $O((\frac{n}{\alpha} + k)f(k, n))$ worst-case time 
or in $O(\frac{n}{\alpha} + kf(k, n))$ expected time
with $n\log\sigma + O(k\log n)$ bits of space, 
where $k$ is the number of sampled positions.
At any moment during the construction,
pattern matching queries can be supported in
$O(\frac{m}{\alpha}f(k,n))$ worst-case time
or in $O(\frac{m}{\alpha} + f(k,n))$ expected time,
where $m$ is the length of the pattern.
\end{theorem}

\subsubsection{Online computation of LZ-Double factorization.}
The \emph{LZ-Double factorization}~\cite{GotoBIT15} (\emph{LZDF})
is a generalization of the Lempel-Ziv 78 factorization~\cite{LZ78}.
The $i$th factor $g_i = g_{i_1} g_{i_2}$ of the LZDF factorization of
a string $T$ of length $n$ is the concatenation of 
previous factors $g_{i_1}$ and $g_{i_2}$ such that
$g_{i_1}$ is the longest prefix of $T[1+\sum_{j=1}^{i-1}|g_j|,n]$ 
that is a previous factor (one of $\{g_1, \ldots, g_{i-1}\} \cup \Sigma$),
and $g_{i_2}$ is the longest prefix of $T[1+|g_{i_1}|+\sum_{j=1}^{i-1}|g_j|,n]$ 
that is a previous factor.
Goto et al.~\cite{GotoBIT15} proposed
a Patricia-tree based algorithm which
computes the LZDF of a given string $T$ of length $n$
in $O(k(M + \min\{k, M\} \log \sigma))$ worst-case 
time\footnote{Since $kM \geq n$ always hods, the $n$ term is hidden in the time complexity.}
with $O(k \log n) = O(n \log \sigma)$ bits of space\footnote{Since all the factors of the LZDF are distinct, $k = O(\frac{n}{\log_\sigma n})$ holds~\cite{LZ78}.},
where $k$ is the number of factors and
$M$ is the length of the longest factor.
Using our packed c-trie,
we can achieve a good expected performance:
\begin{theorem} \label{theo:LZDF}
Using our packed c-trie,
we can compute the LZDF of string $T$ 
in $O(k(\frac{M}{\alpha} + f(k, n)))$ expected time 
with $O(n \log \sigma)$ bits of space.
\end{theorem}

\newcommand{\cellcenter}[1]{\multicolumn{1}{|c||}{#1}}
\newcommand{\cellsinglecenter}[1]{\multicolumn{1}{|c|}{#1}}

\section{Experiments}

In this section, we show our experimental results that compared 
our implementations of the packed c-trie against that of the classical c-trie (Patricia tree).
In Table~\ref{tab:exp:dataset}, we show the datasets and their statistics used in our experiments, where the first six datasets were from Pizza\&Chili Corpus\footnote{Pizza\&Chili Corpus, http://pizzachili.dcc.uchile.cl},  the seventh one consists of URLs in \texttt{uk} domain\footnote{Laboratory for webalgorithmics, uk-2005.urls.gz, http://law.di.unimi.it/datasets.php}, and the eighth one consists of all titles from Japanese Wikipedia\footnote{jawiki, https://dumps.wikimedia.org/jawiki/}. The datasets were treated as binary.

\begin{table}[bt]
 \caption{Description of the datasets}
 \label{tab:exp:dataset}
 \begin{center} 
  \begin{tabular}{|l||r|r|r|r|r|} 
   \hline
   \cellcenter{Data set}
   & \shortstack{Original\\ alhpabet~size}
   & \shortstack{Actual\\ alphabet~size}
   &  \shortstack{Total size \\ (byte)}
   & \shortstack{Number of \\ strings}
  & \shortstack{Ave. string  \\ length (bit)} \\ 
   \hline
   DNA  & 4 & 2 & 52,428,800 & 337 & 1,244,600.59  \\ 
   DBLP  & 128 & 2 &  52,428,800 & 3,229,589 & 129.87  \\ 
   english  & 128 & 2 &  52,428,800 & 9,400,185 & 44.62 \\ 
   pitches  & 128 & 2 &  52,428,800 & 93,354 & 4,492.90  \\ 
   proteins & 20 & 2 &  52,428,800 & 186,914 & 2,243.98  \\ 
   sources & 128 & 2 &  52,428,800 & 5,998,228 & 69.93  \\ 
   urls & 128 & 2 &  52,010,031 & 707,658 & 587.97  \\ 
   jawiki & $\ge 2^{16}$ & 2 &  30,414,297 & 1,643,827 & 148.02  \\ 
   \hline
  \end{tabular}
 \end{center}
\end{table}



We used three implementations of c-tries over the binary alphabet by the authors: an implementation $\alg{CT}$ of classical c-tries, and two simplified implementations $\alg{PCT_{xor}}$ and $\alg{PCT_{hash}}$ of our packed c-tries in Section~\ref{sec:algo} as a proof-of-concept versions. The machine word length $\alpha$ is 32 bits. The details are as follows: 
$\alg{PCT_{xor}}$ only uses the XOR-based technique of Theorem~\ref{lem:large:LCP}, and branching nodes are processed as in the classical c-tries.
$\alg{PCT_{hash}}$ is a simplified implementation of our packed c-tries of Lemma~\ref{lem:large:hash} using hashing. It is equipped with hash tables for $\alpha$-bits integers\footnote{For hash tables, we used the \it{unordered\_map} in C++/STL library.}, but without predecessor/successor data structures. 

We compiled all programs with gcc 4.9.3 using -O3 option, and ran all experiments on a PC (2.8GHz Intel Core i7 processor, register size 64 bits, 16GB of memory) running on MacOS X 10.10.5, where consecutive $\alpha =32$ bits of texts were packed into a machine word.  
For each dataset, we measured the following parameters: 
the number of nodes in the constructed c-trie (\textit{Tree size}),  
the total construction time for the c-trie (\textit{Construction time}), and 
the total time of pattern matching queries (\textit{Query time}).
In the last experiments, pattern strings are consist of  the dataset used for construction. 



\begin{table}[!t]
 \caption{The summary of experimental results}
 \label{tab:exp:result}
 \begin{center}\tabcolsep=0.5em   
  \begin{tabular}{|l||r|r|r|r|r|r|r|r|r|r|r|} 
   \hline
   &\multicolumn{3}{|c|}{Tree size (\# of nodes)}
   &\multicolumn{3}{|c|}{Construction time (msec)}
   &\multicolumn{3}{|c|}{Query time (msec)} \\ 
   \cline{2-10}
   \multicolumn{1}{|c||}{Data set}
   & \multicolumn{1}{|c|}{$\alg{CT}$ }
   & \multicolumn{1}{|c|}{$\alg{PCT_{xor}}$}
   & \multicolumn{1}{|c|}{$\alg{PCT_{hash}}$}
   & \multicolumn{1}{|c|}{$\alg{CT}$} &\multicolumn{1}{|c|}{$\alg{PCT_{xor}}$}  & \multicolumn{1}{|c|}{$\alg{PCT_{hash}}$}  &\multicolumn{1}{|c|}{$\alg{CT}$}  &\multicolumn{1}{|c|}{$\alg{PCT_{xor}}$}  &\multicolumn{1}{|c|}{$\alg{PCT_{hash}}$} \\    
   \hline
   DNA & 674 & 674 & 985 &  \textbf{14,494} & 15,270 & 18,596 & 6,690 & 7,381 &  \textbf{5,342}    \\ 
   DBLP & 1,059,656& 1,059,656 & 1,204,651 & 16,662 & 16,987 &  \textbf{14,139} & 8,083 & 8,905 & \textbf{7,209} \\ 
   english & 448,379& 448,379 & 532,750 & 17,496 &  \textbf{16,944} &  18,197 & \textbf{9,127} & 9,916 & 10,452\\ 
   pitches & 86,205& 86,205 & 121,943 & 18,816 & 16,571 &  \textbf{16,520} & 7,022 & 9,009 & \textbf{6,053}\\ 
   proteins & 310,392& 310,392 & 437,768 & 17,957 & \textbf{15,733} &  18,673 & 8,511 & 8,851  & \textbf{6,749}\\ 
   sources & 1,314,571& 1,314,571 & 1,616,872 & 17,398 & \textbf{15,929} &  16,892 & 8,111 & 8,444 & \textbf{7,852}\\ 
   urls & 1,341,200& 1,341,200 & 1,357,730 & 14,038 & \textbf{13,422} &  13,585 & 6,939 & 6,903 & \textbf{5,918}\\ 
   jawiki & 2,365,821& 2,365,821 & 3,043,817 & 9,440 &  \textbf{9,116} & 10,107 & 4,477 & 4,661 & \textbf{3,962}\\ 
   \hline
  \end{tabular}
 \end{center}
\end{table}



In Table~\ref{tab:exp:result}, we show our experimental results.
First, we consider the first groups of columns on tree size. 
We observed that the number of nodes of \alg{$\alg{PCT_{hash}}$} increases from
both of $\alg{CT}$ and \alg{$\alg{PCT_{xor}}$}.
The gain varies from $101.3\%$ on urls to $146.1\%$ on DNA.
This comes from the addition of boundary nodes.
Next, we consider the second groups of columns on construction time.
We observed that \alg{$\alg{PCT_{xor}}$} is slightly faster than the classical $\alg{CT}$ in most case.
The construction time of \alg{$\alg{PCT_{hash}}$} is slightly faster against $\alg{CT}$ for DBLP, pitches, sources and urls,
and slower for DNA, english, proteins and jawiki.
Yet, the construction time of \alg{$\alg{PCT_{hash}}$} per node is faster than $\alg{CT}$ for all datasets.
We, however, did not observe clear advantage of \alg{$\alg{PCT_{hash}}$} over \alg{$\alg{PCT_{xor}}$}.
We guess that these inconsistency comes from the balance of utility and overhead of creating boundary nodes that depends on datasets.
Finally, we consider the third groups of columns on query time.
Among all datasets except {english}, \alg{$\alg{PCT_{hash}}$} is clearly faster than $\alg{CT}$, where the former achieved $5\%$ to $20\%$ speed-up over the latter. This indicates that \alg{$\alg{PCT_{hash}}$} is superior to the classic c-tries in prefix search.

Overall, we conclude that one of our packed c-trie implementation \alg{$\alg{PCT_{hash}}$} achieved clear speed-up over the classical c-trie implementation in query time for most datasets. In construction time, there seems room of improvements for reducing the overhead of node and hash table creation. 









\bibliographystyle{abbrv}


\newpage
\appendix
\section{Appendix}

In this appendix, we show some proofs
which were omitted due to lack of space.

\subsection{Proof of Lemma~\ref{lem:faster_insert}}

\begin{proof}
Recall the algorithm of Lemma~\ref{lem:large:LCP}.
Starting at the beginning position $i$ in $T$,
we extract the factors of the factorization of $X$ w.r.t. 
the given initial locus $\phi$ on the fly, one by one and from left to right.
We stop the procedure as soon as we find the first mismatch.
Then, we create a new leaf for the inserted string.
The label of the edge leading to the new leaf 
is a pair of positions in $T$, which can be easily computed in $O(1)$ time.
Clearly this gives the desired bounds.
\end{proof}

\subsection{Proof of Theorem~\ref{thm:sparse_suffix_tree}}

\begin{proof}
We explain how we can build the sparse suffix trees of~\cite{Karkkainen:Ukkonen:1996} efficiently.
For an integer parameter $r > 1$,
K\"arkk\"ainen and Ukkonen's algorithm (KU-algorithm, in short)~\cite{Karkkainen:Ukkonen:1996} constructs the \emph{$r$-evenly sparse suffix tree} 
of the input string $T$.
KU-algorithm differs from Ukkonen's online suffix tree construction algorithm 
in that KU-algorithm uses 
\emph{$r$-letter suffix links},
such that the suffix link of each node $v$ is a pointer to
the node $u$ such that $\Str(u) = \Str(v)[r+1..|\Str(v)|]$,
but otherwise is the same as Ukkonen's algorithm.
This results in a compact trie which stores 
the evenly-spaced $\lfloor n /r \rfloor + 1$ suffixes $T[1,n]$, $T[1+r,n]$, 
\ldots, $T[1+r\lfloor n / r \rfloor,n]$ of $T$.

KU-algorithm scans the input string $T$ from left to right,
and when the algorithm processes the $i$th letter of $T$,
the $r$-evenly sparse suffix tree of $T[1,i]$ is maintained.
This is done by inserting the leaves into the current compact trie
in increasing order of the positions the leaves correspond to.
Assume that while processing the $i$th letter of $T$,
the algorithm has just inserted the $j$th leaf $\ell_j$ for 
sampled position $1+(j-1)r$ of $T$. 
If the suffix $T[1+jr,i]$ of $T[1,i]$ is not recognized by the current compact trie,
then the algorithm inserts the $(j+1)$th leaf $\ell_{j+1}$ 
for the next sampled position $1+jr$.
This can be done as follows:
For any node $v$, let $\rsl(v)$ denotes the $r$-letter 
suffix link of $v$.
Let $v_j$ be the nearest ancestor of $\ell_j$ for which 
$\rsl(v_j)$ is already defined
($v_j$ is either $\parent(\ell_j)$ or $\parent(\parent(\ell_j))$).
We follow the suffix link and let $u_{j+1} = \rsl(v_j)$.
Let $\phi_{j+1}$ be the locus of $\Str(u_{j+1})$,
namely $\phi_{j+1} = (e, |e|)$ with $e = (\parent(u_{j+1}), u_{j+1})$.
Let $X_{j+1} = T[i-h+1, i]$, where 
$h = |T[j+1,i]| - |\Str(\phi_{j+1})| = i - j - |\Str(\phi_{j+1})|$.
The leaf $\ell_{j+1}$ can be added to the compact trie
by inserting the string $X_{j+1}$ from the locus $\phi_{j+1}$.

We apply our micro-trie decomposition to the sparse suffix tree,
and use our techniques in Section~\ref{sec:algo} and 
in Lemma~\ref{lem:faster_insert}.
Then, the total time complexity to construct the $r$-evenly sparse suffix tree
of $T$ is proportional to the amount of work 
of the $\Insert$ operations of Lemma~\ref{lem:faster_insert} for all leaves.
For each $1 \leq j \leq k$ 
let $q_{j}$ be the length of the longest prefix of $X_{j}$
that can be spelled out from $\phi_{j}$.
Now we estimate $\sum_{j=1}^{k} \frac{q_{j}}{\alpha}$.
Each time we traverse an $r$-letter skipping suffix link,
the string depth decreases by $r$.
Since $k = \lfloor n /r \rfloor + 1$ and we traverse $r$-letter suffix links exactly 
$k-1$ times,
we can conclude that $\sum_{j=1}^{k}q_{j} = O(n)$,
which implies that $\sum_{j=1}^{k} \frac{q_{j}}{\alpha} = O(n / \alpha)$.
Since we perform $\Insert$ operations exactly $k$ times,
the $r$-evenly sparse suffix tree can be constructed in 
$O((\frac{n}{\alpha} + k) f(k,n))$ worst-case time
or in $O(\frac{n}{\alpha} + kf(k,n))$ expected time.

The bounds for word suffix trees of Inenaga and Takeda~\cite{Inenaga:Takeda:CPM:2006}
and those of suffix trees on variable-length codes of 
Uemura and Arimura~\cite{Uemura:Arimura:CPM:2011} can be obtained similarly.
\end{proof} 

\subsection{Proof of Theorem~\ref{theo:LZDF}}

\begin{proof}
Suppose we have computed the first $j-1$ factors
$g_1, \ldots, g_{j-1}$ and we are now computing the $j$th factor $g_j$.
We store the previous factors $g_1, \ldots, g_{j-1}$
in our packed c-trie.
In addition, for any previous factor $g_{i}$~($1 \leq i < j$),
if there is no leaf or branching node which represents $g_{i}$,
then we add an internal non-branching node for $g_{i}$
into the packed c-trie.
We mark only and all nodes which represent previous factors.
To compute the $j$th factor $g_j = g_{j_1}g_{j_2}$,
we perform $\LPath(r, T_j)$ query 
where $r$ is the locus for the root and 
$T_j = T[1+\sum_{i}^{j-1}|g_i|,n]$.
Let $\hat{\phi}$ be the answer to the query.
Note that $\hat{\phi}$ can be deeper than the locus for $g_{j_1}$,
but it is always in the subtree rooted at $g_{j_1}$.
Hence, the nearest marked ancestor (NMA) of $\hat{\phi}$ is $g_{j_1}$.
We can compute $g_{j_2}$ similarly.
After we computed $g_j$,
we perform $\Insert(r, g_{j})$ operation and then 
mark the node which represents $g_{j}$.

The depth of the locus $\hat{\phi}$ 
is bounded by the length $M$ of the longest factor.
Hence we can reach the locus $\hat{\phi}$ in $O(\frac{M}{\alpha} + f(k, n))$ 
expected time using our packed c-trie.
We repeat the above procedure $k$ times.
Using the semi-dynamic NMA data structure of Westbrook~\cite{westbrook92:_fast_increm_planar_testin} that
supports NMA queries, inserting new nodes,
and marking unmarked nodes in amortized $O(1)$ time each,
we obtain the desired bound.
\end{proof}

\end{document}